\documentclass[11pt]{article}
\usepackage[letterpaper]{geometry}
\renewcommand{\baselinestretch}{1.59}\large\normalsize 
\renewcommand{\baselinestretch}{1.0}\large\normalsize 
\usepackage{amsbsy}
\usepackage{amsmath}
\usepackage{amsthm}
\usepackage{amssymb}
\usepackage{amsfonts}
\usepackage{times}
\usepackage{graphicx}
\usepackage{color}
\usepackage{authblk}
\usepackage{jhtitle}
\usepackage{enumitem}

\usepackage[scaled=0.90]{helvet}
\usepackage{ifthen}
\usepackage{xcomment}
\usepackage{versions}
\usepackage{natbib}
\newtheoremstyle{mytheorem}{}{}{\slshape}{}{\bfseries}{}{2mm}{}
\theoremstyle{mytheorem}
\newtheorem{remark}{Remark}
\newtheorem{theorem}{Theorem}

\newtheorem{definition}{Definition}

\newtheorem{lemma}{Lemma}

\newcommand{\commentt}[1]{}

\newcommand{\real}{\mathbb{R}}

\newcounter{itemnum}

\newboolean{include-notes-to-stay}
\setboolean{include-notes-to-stay}{false}

\ifthenelse{\boolean{include-notes-to-stay}}%
  {\newenvironment{notes-to-stay}
  {\begin{color}{blue}}{\end{color}}}%
  {\newxcomment[]{notes-to-stay}}

\newboolean{include-notes2}
\setboolean{include-notes2}{true}

\ifthenelse{\boolean{include-notes2}}%
  {}%
  {\newxcomment[]{notes2}}

\newboolean{include-notes-aixin}
\setboolean{include-notes-aixin}{false}

\ifthenelse{\boolean{include-notes-aixin}}%
  {\newenvironment{notes-aixin}
  {\begin{color}{magenta}}{\end{color}}}%
  {\newxcomment[]{notes-aixin}}

\numberwithin{equation}{section}
\newcommand{\NM}{N_{\max}}
\newcommand{\admiss}{\mathcal{N}} 

\begin{document}

\title{Optimal and Efficient Sample Size Re-estimation: A Dynamic Cost Framework}

\author[1]{Rui Jin}
\author[2]{Cai Wu}
\author[2]{Qiqi Deng}
\affil[1]{Novartis Pharmaceuticals Corporation, East Hanover, NJ, USA}
\affil[2]{Moderna Inc., Cambridge, MA, USA}

\date{ }


\maketitle
\renewcommand{\baselinestretch}{1.42}\large\normalsize

\begin{abstract}
Adaptive sample size re-estimation (SSR) is a well-established strategy for improving the efficiency and flexibility of clinical trials. Its central challenge is determining whether, and by how much, to increase the sample size at an interim analysis. This decision requires a rational framework for balancing the potential gain in statistical power against the risk and cost of further investment. Prevailing optimization approaches, such as the Jennison and Turnbull (JT) method, address this by maximizing power for a fixed cost per additional participant. While statistically efficient, this paradigm assumes the cost of enrolling another patient is constant, regardless of whether the interim evidence is promising or weak. This can lead to impractical recommendations and inefficient resource allocation, particularly in weak-signal scenarios.

%

We reframe SSR as a decision problem under dynamic costs, where the effective cost of additional enrollment reflects the interim strength of evidence. Within this framework, we derive two novel rules: (i) a likelihood-ratio based rule, shown to be Pareto optimal in achieving smaller average sample size under the null without loss of power under the alternative; and (ii) a return-on-investment (ROI) rule that directly incorporates economic considerations by linking SSR decisions to expected net benefit. To unify existing methods, we further establish a representation theorem demonstrating that a broad class of SSR rules can be expressed through implicit dynamic cost functions, providing a common analytical foundation for their comparison. Simulation studies calibrated to Phase III trial settings confirm that dynamic-cost approaches improve resource allocation relative to fixed-cost methods.

\end{abstract}


\renewcommand{\baselinestretch}{1.02}\large\normalsize
\vfil %

\thanks{
\noindent %
}

\renewcommand{\baselinestretch}{1.42}\large\normalsize
\thispagestyle{empty} %
\eject %

\setcounter{page}{1}
\section{Introduction}
\label{sec:intro}

In the design of clinical trials, determining an appropriate sample size is a fundamental challenge, especially when preliminary information on the treatment effect is limited. An undersized trial may lack sufficient statistical power to detect a true treatment effect, risking a failed study, whereas an oversized trial is inefficient, consuming excessive resources and exposing more participants than necessary to a potentially suboptimal intervention.

To mitigate the uncertainty inherent in sample size determination, adaptive designs are now widely adopted. A common approach is the group sequential design~\citep{pocock1977group}, which allows for early stopping at pre-planned interim analyses. A more flexible alternative is sample size re-estimation, which permits adjustment of the final sample size based on interim data. SSR can be conducted in a blinded or unblinded manner~\citep{friede2013blinded, proschan2005two}. Within the class of unblinded SSR designs, the promising zone approach has garnered substantial attention. This method modifies the sample size based on the observed interim test statistic to achieve the target power. Extensive literature has examined methods to control the type I error rate, define the promising zone, and allocate additional participants effectively. Notably, many proposed promising zone designs can be shown to yield equivalent sample size rules through careful specification of design parameters~\citep{mehta2022optimal,hsiao2019optimal,pilz2021optimal}.

A key contribution to this field is the framework of Jennison and Turnbull (JT)~\citep{jennison2015adaptive}, who formulated SSR as an optimization problem balancing conditional power against the expected number of additional participants. The JT approach has strong theoretical appeal, yielding designs that minimize the expected sample size among all designs achieving a prespecified power. A limitation, however, is its reliance on a fixed per-patient cost, which implicitly assumes that the burden of additional enrollment is constant regardless of the strength of interim evidence. In practice, this assumption can lead to inefficient or counterintuitive recommendations, such as substantial sample size increases even when interim data provide little support for efficacy.

To address this limitation, we propose a generalized SSR framework built upon dynamic cost functions, where the effective cost of an additional participant adapts to the strength of the interim evidence. Within this framework, we introduce two novel, optimal decision rules. The first is a likelihood-ratio based rule, which we prove is Pareto optimal. This property ensures it provides a highly efficient and well-balanced solution to the multi-objective problem of maximizing power while controlling the average sample size under both the null and alternative hypotheses. Its dynamic cost function, which incorporates the likelihood ratio, systematically penalizes sample size increases for interim data that favor the null hypothesis. The second is a return-on-investment (ROI) based rule that directly integrates economic considerations into the design, linking SSR decisions to the trial's expected net benefit. The ROI-based approach demonstrates how a re-estimation rule can be calibrated to reflect different cost-benefit profiles and prior beliefs about treatment effectiveness, thereby allocating resources where the expected payoff is highest. We show that the dynamic cost function for this rule is a Bayesian normalization of the likelihood-ratio approach, establishing a clear theoretical link between these two innovative rules.

A central theoretical contribution of this work is a representation theorem for this generalized framework. This theorem establishes that if the conditional power function satisfies certain regularity conditions, there exists a one-to-one correspondence between an SSR rule and its implicit dynamic cost function. Consequently, a broad class of SSR rules can be reverse-engineered to reveal their underlying cost structures, providing a unified foundation for their analysis and comparison. As an application, we derive the implicit cost function for a constrained promising zone (CPZ) rule and conduct a rigorous comparison with the JT rule. This analysis reveals that the CPZ rule implicitly utilizes a dynamic cost that concentrating the allocation of additional resources at the turning point of its promising zone.

Simulation studies were conducted to evaluate the operating characteristics of our proposed methods against established benchmarks. The results demonstrate that the likelihood-ratio based rule is more conservative than the JT rule, recommending smaller sample size increases when interim data are more consistent with the null hypothesis or when the interim analysis is performed with limited data. This highlights the rule's ability to adapt not just to the interim effect estimate but also to its precision. Furthermore, the simulations confirm that the ROI-based rule effectively incorporates prior information and pre-specified cost-benefit profiles to tailor SSR decisions, a capability not present in the other designs considered.

The remainder of this paper is organized as follows. Section~\ref{sec:bk} reviews the JT framework. Section~\ref{sec:method} introduces our generalized framework and the two specific rules derived from it. Section~\ref{sec:rep} presents the representation theorem for sample size re-estimation rules. Section~\ref{sec:simulation} details the results from extensive simulation studies, and Section~\ref{sec:end} concludes with a discussion of the implications of our findings.

\section{Background of the JT Framework}\label{sec:bk}
We provide a brief introduction to the sample size re-estimation method developed by~\cite{jennison2015adaptive}, which aims to address the challenge of determining an optimal sample size increase after observing interim results.

We consider a clinical trial with a continuous endpoint, where measurements for the treatment and control arms are denoted by $Y_i^{trt}$ and $Y_i^{ctl}$, $i = 1, 2, \dots$, respectively. Assuming normality with a common known variance $\sigma^2$, we have:
\begin{equation*}
Y_i^{trt} \sim N(\mu_t, \sigma^2) \hspace{.5mm} \text{ and } \hspace{.5mm} 	Y_i^{ctl} \sim N(\mu_c, \sigma^2) \,.
\end{equation*}
The treatment effect $\theta$ is defined as the difference in means $\theta = \mu_t - \mu_c$. An interim analysis is planned after collecting $n_1$ measurements with a 1:1 allocation ratio. The observed treatment effect is $\hat{\theta}_1 = \bar{Y}^{trt} - \bar{Y}^{ctl}$ and the corresponding standardized test statistic is denoted by $Z_1$. 


The final decision uses an inverse-normal combination test with weights $w_1,w_2>0$ satisfying $w_1^2+w_2^2=1$ and critical value $C_{\mathrm{crit}}$. Conditional on $Z_1=z_1$ and total final sample size $n_2\ge n_1$, the conditional power under effect size $\theta$ can be written in the standard form
\begin{equation}\label{eq:CP}
CP_\theta(z_1,n_2)\;=\;\Phi\!\Big(A(z_1,n_2)\Big),\quad
A(z_1,n_2)\;=\;K\sqrt{n_2-n_1}\;-\;c(z_1),
\end{equation}
with $K=\theta/(2\sigma)$ and $c(z_1)=(C_{\mathrm{crit}}-w_1 z_1)/w_2$.

\cite{jennison2015adaptive} developed an optimization framework to determine the optimal total sample size, $n_2(z_1)$, for the final analysis after observing the interim statistic $Z_1=z_1$. This framework defines the optimal sample size re-estimation rule, $n_2^{JT}(z_1)$, as the solution to the following maximization problem:
\begin{equation}\label{JT}
n_2^{JT}(z_1) =\max_{n_2\in[n_{\min},N_{\max}]}\big\{CP_\theta(z_1,n_2)-\gamma\,(n_2-n_{\min})\big\}, \,,
\end{equation}
where $\gamma>0$ is a fixed cost parameter and $[n_{\min},N_{\max}]$ is the allowable range of total sample sizes. The term $CP_{{\theta}}(z_1, n_2(z_1))$ represents the conditional power, given $Z_1=z_1$, evaluated under a prespecified effect size of interest, ${\theta}$. The quantity $\gamma$ serves as a tuning parameter that represents the acceptable ``rate of exchange" between conditional power and sample size. It can be tied to the cost of additional subjects and is  used to control the extent to which the sample size may be increased based on the interim data.

A key advantage of this optimization framework is that a sample size rule that maximizes this conditional objective for every observed $z_1$ also possesses a desirable unconditional optimality property. Specifically, this rule minimizes the expected sample size, $E_{{\theta}}(N)$, for all designs that achieve the same overall power under the specified effect size, ${\theta}$. This framework provides a principled approach to determining sample size increases, balancing the desire for higher power with the cost of additional observations in a consistent manner.

\section{Optimal Sample Size Re-estimation with Dynamic Cost Functions}\label{sec:method}

In this section, we introduce a family of sample size re-estimation decision rules that utilize \textbf{dynamic cost functions}. This approach can be viewed as a generalization of the JT framework. The JT rule only adds a sample if the marginal gain in conditional power exceeds a constant threshold, $\gamma$. This constant threshold makes the decision to add samples heavily dependent on the property of the conditional power alone, which might not directly reflect the magnitude of the interim test statistic. Consequently, in regions where the observed interim test statistic, $z_1$, indicates a weak signal, the JT rule may inappropriately add a large number of samples to boost the conditional power. This decision is often not practical in real-world scenarios.

To mitigate this shortcoming and enable a more flexible allocation of sample size, we propose a generalized framework for the JT method that incorporates a dynamic cost function. Specifically, we allow the cost of adding a sample to depend on the interim test statistic, $z_1$. Therefore, we seek a rule that solves the following optimization problem:
\begin{equation}\label{dynamic}
\arg\max_{n_2\in[n_{\min},N_{\max}]}\left\{ CP_{\theta}(z_1, n_2) - \gamma(z_1) n_2 \right\} \,.
\end{equation}
The JT method can be considered a special case of this framework where $\gamma(z_1)$ is a constant function. In the following sections, we will present two specific forms of $\gamma(z_1)$ that address the limitations of the JT method and offer new interpretations for sample size re-estimation.

\subsection{Likelihood-ratio based  $\gamma(z_1)$}\label{sec:LR_gamma}
Let us consider a trial with treatment effect denoted by $\theta$, where the null hypothesis is $\Theta_0 = \{0\}$ and the alternative is $\Theta_1 = \{\theta\}$, with $\theta > 0$. We are interested in a sample size re-estimation rule $n_2(z_1)$ that solves the following optimization problem:
\begin{equation}\label{prob_double_constrain}
\max_{n_2 \in \admiss} P_{\theta} (\mathrm{Reject}\,H_0 | \Theta_1) \hspace{0.5cm} \text{subject to} \hspace{0.5cm} E_0 [n_2(Z_1)] \leq B_0, \hspace{0.5cm} \text{and} \hspace{0.5cm} E_{\theta} [n_2(Z_1)] \leq B_{\theta} \,,
\end{equation}
where $B_0, B_{\theta} > 0$ are predefined budget constraints. The goal is to maximize the unconditional power while controlling the average sample size under both the null and alternative hypotheses.

We let $f_0(z_1)$ and $f_\theta(z_1)$ denote densities of $Z_1$ under $\Theta_0$ and $\Theta_1$, respectively. We assume that $f_\theta(z_1) > 0 $ for $z_1 \in \real$. For nonnegative Lagrange multipliers $\lambda_1,\lambda_2\ge 0$, the Lagrangian is
\[
\mathcal{L}\!\left(n_2(\cdot),\lambda_1,\lambda_2\right)
= \int \Big\{ CP_\theta\!\big(z_1, n_2(z_1)\big)\, f_\theta(z_1)
- \big(\lambda_1 f_0(z_1)+\lambda_2 f_\theta(z_1)\big)\, n_2(z_1) \Big\}\, \mathrm{d}z_1 .
\]
Because the integral is separable in $z_1$, that is, $n_2$ enters the objective only through its value at the same $z_1$ and there are no cross $z_1$ couplings (no terms involving $n_2(z_1')$ with $z_1'\neq z_1$), maximizing $\mathcal{L}$ over  $n_2(\cdot)$ reduces to pointwise maximization of the integrand. On the set where $f_\theta(z_1)>0$, dividing by $f_\theta(z_1)$ does not change the maximizer and yields the equivalent problem
\[
\max_{\,n_2(z_1)} \;\Big\{ CP_\theta\big(z_1,n_2(z_1)\big)
- \Big(\lambda_2 + \lambda_1 \frac{f_0(z_1)}{f_\theta(z_1)}\Big)\, n_2(z_1) \Big\},
\]
i.e., a pointwise trade-off between conditional power and a dynamic cost term proportional to the likelihood ratio $f_0(z_1)/f_\theta(z_1)$.
We denote the solution of this optimization problem by $n_2^{LR}(z_1)$. This formulation yields a dynamic cost function in the JT fashion, where the cost is given by $\gamma(z_1) = \lambda_2 + \lambda_1 \frac{f_0(z_1)}{f_{\theta}(z_1)} $. The term $\frac{f_0(z_1)}{f_{\theta}(z_1)}$ is the likelihood ratio. A large value of this ratio indicates that the observed interim statistic $z_1$ is more likely to arise under the null hypothesis, thereby increasing the cost of adding a sample and making it harder to justify a sample size increase.

\subsubsection{Comparing  $n_2^{JT}$ with $n_2^{LR}$}
Let $n_2^{\mathrm{JT}}:\real \to\mathcal N$ denote the JT rule and $n_2^{\mathrm{LR}}:\real \to\mathcal N$ the proposed likelihood ratio based rule; we write $n_2^{\mathrm{JT}}(z_1)$ and $n_2^{\mathrm{LR}}(z_1)$ for their values at a given $z_1$. For a fixed cost parameter $\gamma$, we write $n_2^{\mathrm{JT}}(\cdot;\gamma)$.
We now formally compare the performance of the JT rule, $n_2^{JT}$, with our proposed likelihood-ratio based rule, $n_2^{LR}$. For a fixed cost parameter $\gamma$, let $n_2^{JT}(\cdot, \gamma)$ be the optimal rule from the JT method. We claim that $n_2^{JT}(\cdot, \gamma)$ is not the optimal solution for the following constrained optimization problem:
\begin{equation*}
\min_{n_2} E_{0} [n_2(Z_1)] \quad \text{subject to} \quad P_{\theta} (\mathrm{Reject}\,H_0; n_2) = P_{\theta} (\mathrm{Reject}\,H_0; n_2^{JT}(\cdot, \gamma)) \quad \text{and} \quad E_{\theta} [n_2(Z_1)] = E_{\theta} [n_2^{JT}(Z_1, \gamma)] \,.
\end{equation*}
This problem seeks the rule that minimizes the average sample size under the null hypothesis while maintaining the same unconditional power and average sample size under the alternative as the JT rule.

For non-zero Lagrange multipliers $\mu_1$ and $\mu_2$, the Lagrangian function for this problem is:
\begin{align*}
\mathcal{L}(n_2, \mu_1, \mu_2) &= E_0[n_2(Z_1)] - \mu_1 \left( P_{\theta} (\mathrm{Reject}\,H_0; n_2) - P_{\theta} (\mathrm{Reject}\,H_0; n_2^{JT}(\cdot, \gamma))\right) \\
&\quad - \mu_2 \left(E_{\theta} [n_2(Z_1)] - E_{\theta} [n_2^{JT}(Z_1)] \right)
\end{align*}
For a function $n_2$ to be a minimizer of $\mathcal{L}(n_2, \mu_1, \mu_2)$, the integrand must be minimized for each $z_1$. We treat $n_2$ as a continuous variable for the purpose of optimization, a standard approach for this class of problems. While the sample size must be an integer, treating it as a continuous variable allows for the use of calculus-based optimization. The resulting continuous solution can then be rounded to the nearest integer. The performance loss from this approximation is typically negligible. Thus, a necessary condition for $n_2$ to be a minimizer is:
\begin{equation*}
\frac{\partial}{\partial n_2(z_1)} \left\{n_2(z_1)f_0(z_1) - \mu_1 CP_{\theta}(z_1, n_2(z_1))f_{\theta}(z_1) - \mu_2 n_2(z_1) f_{\theta}(z_1) \right\} = 0 \,.
\end{equation*}
which simplifies to:
\begin{equation*}
\frac{\partial CP_{\theta}(z_1, n_2(z_1))}{\partial n_2(z_1)} = \frac{\mu_2}{\mu_1} + \frac{1}{\mu_1} \frac{f_0(z_1)}{f_{\theta}(z_1)} \,.
\end{equation*}
For the JT method, the corresponding condition is $\frac{\partial CP_{\theta}(z_1, n_2(z_1))}{\partial n_2(z_1)} = \gamma$, which is a constant. Since the right-hand side of our derived condition is a function of $z_1$, this proves that $n_2^{JT}(z_1)$ is not the rule that minimizes $E_0[N]$ for its given power and average sample size under the alternative hypothesis.

By setting $\lambda_1 = \frac{1}{\mu_1}$ and $\lambda_2 = \frac{\mu_2}{\mu_1}$, we can see that our extended framework can provide a sample size re-estimation rule that is better than the JT method in the sense that it yields a smaller average sample size under the null hypothesis for a comparable power and average sample size under the alternative.

\subsubsection{Pareto optimality of $n_2^{LR}$}
The likelihood-ratio based re-estimation rule, $n_2^{LR}$ is a solution to multi-objective optimization problem. As such, it possesses a fundamental property known as Pareto optimality.

\begin{definition}
	A rule $n_2^*$ is Pareto optimal if no other rule $n_2$ exists that is strictly better on at least one objective while being no worse on the others.
\end{definition}

\begin{lemma}
	$n_2^{LR}$ is Pareto optimal. 
\end{lemma}
\begin{proof}
	For $\lambda_1 \geq 0$ and $\lambda_2 \geq 0$,  $n_2^{LR}$ maximizes $\mathcal{L}(n_2) = P_{\theta} (\mathrm{Reject}\,H_0; n_2) - \lambda_21 E_0 [n_2(Z_1)] - \lambda_2 E_{\theta}[n_2(Z_1)]$. Assume that $n_2^{LR}$  is not Pareto optimal. By definition, this means there exists a rule $n^*_2$ such that
	\begin{enumerate}
		\item[] $P_{\theta} (\mathrm{Reject}\,H_0; n_2^*) > P_{\theta} (\mathrm{Reject}\,H_0; n^{LR}_2) $ \,,
		\item[] $E_0 [n^*_2(Z_1)] \leq E_0 [n^{LR}_2(Z_1)] $ \,,
		\item[] $E_{\theta} [n^*_2(Z_1)] \leq E_{\theta} [n^{LR}_2(Z_1)] $ \,.
	\end{enumerate}
Then, we got $\mathcal{L}(n^*_2) > \mathcal{L}(n^{LR}_2) $, which contradicts to the fact that $n_2^{LR}$ maximizes $\mathcal{L}(n_2) $.
\end{proof}
The Pareto optimality of $n_2^{LR}$ establishes it as an efficient and well-balanced solution for the multi-objective problem of maximizing power while controlling the average sample size under both $\Theta_0$ and $\Theta_1$. 

In the context of sample size re-estimation, Pareto optimality means that the likelihood-ratio based rule avoids wasteful designs. If one attempts to reduce sample size under the null, the only way to do so is by lowering power; conversely, if one increases power, additional participants must be recruited. Thus, $n_2^{LR}$ represents an efficiency frontier: it offers the best achievable balance between power and resource use. This provides reassurance to practitioners that adopting $n_2^{LR}$ does not leave potential efficiency gains unexploited.

\subsection{Return on investiment (ROI) based $\gamma(z_1)$}

Let us consider a cost-benefit framework for sample size re-estimation. Suppose each additional sample has a cost $c > 0$. If the final results declare the treatment effective, the sponsor receives a return $V > 0$. Given an interim test statistic $z_1$ and a final sample size rule $n_2(z_1)$, the expected net return is given by:
\begin{equation}
E[\text{Net Return} | z_1] = V \cdot P(\text{treatment effective and final reject} | z_1, n_2(z_1)) - c n_2^{}(z_1) \,.
\end{equation}
The term $P(\text{treatment effective and final reject} | z_1, n_2^{}(z_1))$ can be factored using the rules of conditional probability:
\begin{equation*}
P(\text{treatment effective and final reject} | z_1, n_2^{}(z_1)) = P(\text{effective } | z_1) \cdot P(\text{final reject} | \text{effective}, z_1, n_2^{}(z_1)) \,.
\end{equation*}
Let $p(z_1) = P(\text{effective} | z_1)$ be the posterior probability of the treatment being effective, and $CP_{\theta}(z_1, n_2^{}(z_1))$ be the conditional power. The expected net return can then be expressed as:
\begin{equation}
E[\text{Net Return} | z_1] = V p(z_1) \cdot CP_{\theta}(z_1, n_2^{}(z_1)) - c n_2^{}(z_1) \,.
\end{equation}
Maximizing this expected return is equivalent to solving a JT-style optimization problem where the cost parameter is a dynamic function of $z_1$, specifically $\gamma(z_1) = c / (V p(z_1))$. This formulation leads to more aggressive sample size increases for stronger interim results (i.e., higher $p(z_1)$) compared to the standard JT method with a fixed $\gamma$. 

We denote the resulting rule for maximizing $E[\text{Net Return} | z_1]$ by $n_2^{ROI}(z_1)$.  This rule is the Bayes optimal solution for maximizing the corresponding expected net return. Being Bayes optimal means that the rule is the best possible decision strategy for the given objective, as it formally accounts for all available information (prior beliefs and observed data) to maximize the expected return. This provides a theoretical foundation for making resource allocation decisions based on $n_2^{ROI}(z_1)$.

We can derive a more explicit form for $\gamma(z_1)$ using a two-state prior model. Let the parameter space be partitioned into two disjoint sets, $\Theta_0$ (no treatment effect) and $\Theta_1$ (effective treatment), with prior probabilities $\pi_0 = P(\theta \in \Theta_0)$ and $\pi_1 = P(\theta \in \Theta_1)$. By Bayes' rule, the posterior probability of effectiveness given $z_1$ is:
\begin{equation*}
p(z_1) = \frac{\pi_1 f_{\theta}(z_1)}{\pi_0 f_0(z_1) + \pi_1 f_{\theta}(z_1)} \,,
\end{equation*}
where $f_0(z_1)$ and $f_{\theta}(z_1)$ are the density functions of the interim test statistic under $\Theta_0$ and $\Theta_1$, respectively. For composite hypotheses, these densities are replaced by the marginal likelihoods.

Substituting this expression into the dynamic cost function yields:
\begin{equation}\label{gamma_roi}
\gamma(z_1) = \frac{c}{V p(z_1)} = \frac{c}{V} \left(1 + \frac{\pi_0 f_0(z_1)}{\pi_1 f_{\theta}(z_1)} \right) \,.
\end{equation}
This result shows that the ROI-based extension is a Bayesian normalization of the likelihood-ratio based approach discussed in Section~\ref{sec:LR_gamma}. This framework connects sample size re-estimation directly to a return on investment  perspective. This ROI alignment property is highly desirable for decision-makers who care not only about statistical power but also about the expected net benefit. A rule derived from this extension allocates resources where the expected payoff is highest, which is both economically rational and clinically appealing.

To aid practitioners, we outline practical guidelines for ROI calibration. The cost per participant $c$ can be estimated directly from trial budgets (including recruitment, treatment, and follow-up costs). The return $V$ may be approximated by the expected incremental revenue or health benefit of a successful approval, discounted for market size and time horizon. Priors $(\pi_0, \pi_1)$ can be elicited from historical data or expert belief about similar compounds in the same therapeutic area. 

\subsubsection{Link to portfolio management}

The ROI framework provides a powerful quantitative tool for portfolio management by translating the statistical outcomes of clinical trials into a standardized financial metric: the Expected Net Return (ENR). This allows for a direct, data-driven comparison of different projects competing for limited company resources. The ENR is inherently risk-adjusted through the posterior probability $p(z_1)$. A project with a massive potential return $V$ will still have a low ENR if the interim data suggests a low probability of success, preventing the company from chasing long shots with poor evidence. Furthermore, this method replaces subjective debate with a transparent, quantitative ranking of projects. It helps defend difficult decisions, such as allocating resources to a riskier project over a seemingly safer one, based on a clear and defensible rationale.

\section{Representation Theorem for Sample Size Re-estimation Rules}\label{sec:rep}
The landscape of SSR methodologies is populated by numerous strategies, including promising zone designs~\citep{hsiao2019optimal,mehta2022optimal} , rules based on achieving a target conditional power, and other procedures. These methods are often described in different terms and justified by different heuristics, making direct comparison difficult. The dynamic cost framework potentially provides a common language and a unified foundation for their analysis. By calculating the implied cost function, $\gamma(z_1)$, for each distinct rule, their underlying assumptions about the value of information and the trade-off between power and sample size can be made explicit. This allows for a rigorous, ``apples-to-apples" comparison of various SSR rules, revealing the true operational characteristics that might be obscured by their procedural definitions. 

We now present a representation theorem for sample size re-estimation rules under several regularity conditions. 

	

\begin{theorem}\label{ssr_rep}
	Let $\mathcal P\subseteq\mathbb R$ denote the promising zone and let $\mathcal N:=(n_{\min},N_{\max})$.
	Suppose that for every $z_1\in\mathcal P$ the map $n_2\mapsto CP_\theta(z_1,n_2)$ is differentiable and strictly concave on $\mathcal N$.
	Let $n_2^*:\mathcal P\to\mathcal N$ be an SSR rule taking values in the interior. Then there exists a unique function
	\[
	\gamma^*:\mathcal P\to\mathbb R,\qquad 
	\gamma^*(z_1):=\left.\frac{\partial}{\partial n_2}CP_\theta(z_1,n_2)\right|_{n_2=n_2^*(z_1)},
	\]
	such that, for every $z_1\in\mathcal P$, the point $n_2^*(z_1)$ is the unique maximizer of
	\[
	\max_{n_2\in\mathcal N}\;\Big\{ CP_\theta(z_1,n_2)-\gamma^*(z_1)\,n_2 \Big\}.
	\]
	Moreover, if $\tilde\gamma:\mathcal P\to\mathbb R$ satisfies 
	$n_2^*(z_1)\in\arg\max_{n_2\in\mathcal N}\{CP_\theta(z_1,n_2)-\tilde\gamma(z_1)n_2\}$ for all $z_1\in\mathcal P$, 
	then $\tilde\gamma=\gamma^*$  on $\mathcal P$.
\end{theorem}

\begin{proof}
	Fix any $z_1\in\mathcal P$ and define 
	\[
	\gamma^*(z_1):=\left.\frac{\partial}{\partial n_2}CP_\theta(z_1,n_2)\right|_{n_2=n_2^*(z_1)}.
	\]
	Consider the one dimensional objective
	\[
	\phi_{z_1}(n_2)\ :=\ CP_\theta(z_1,n_2)-\gamma^*(z_1)\,n_2,\qquad n_2\in\mathcal N.
	\]
	By assumption, $n_2\mapsto CP_\theta(z_1,n_2)$ is strictly concave and differentiable on $\mathcal N$; subtracting
	a linear term preserves strict concavity. Hence $\phi_{z_1}$ is strictly concave on $\mathcal N$.
	Moreover,
	\[
	\frac{\partial}{\partial_{n_2}}\phi_{z_1}(n_2)\ =\ 	\frac{\partial}{\partial_{n_2}}CP_\theta(z_1,n_2)-\gamma^*(z_1),
	\]
	so $	\frac{\partial}{\partial_{n_2}}\phi_{z_1}\big(n_2^*(z_1)\big)=0$ by construction. For a strictly concave, differentiable objective,
	the first order condition is both necessary and sufficient and the maximizer is unique; thus $n_2^*(z_1)$ is the
	unique maximizer of $\phi_{z_1}$ over $\mathcal N$.
	
	For uniqueness of the cost, suppose $\tilde\gamma:\mathcal P\to\mathbb R$ also satisfies
	$n_2^*(z_1)\in\arg\max_{n_2\in\mathcal N}\{CP_\theta(z_1,n_2)-\tilde\gamma(z_1)n_2\}$ for all $z_1\in\mathcal P$.
	Strict concavity implies uniqueness of the maximizer and, hence, the first order condition at $n_2^*(z_1)$:
	$	\frac{\partial}{\partial_{n_2}}CP_\theta(z_1,n_2^*(z_1))=\tilde\gamma(z_1)$. By the definition of $\gamma^*$, we conclude
	$\tilde\gamma(z_1)=\gamma^*(z_1)$ for all $z_1\in\mathcal P$.
\end{proof}

The one-to-one correspondence described in Theorem~\ref{ssr_rep} applies when the SSR rule yields an interior solution. When the rule assigns a boundary value, the equality defining the cost function is replaced by an inequality:
\begin{itemize}
	\item \textbf{Maximum boundary} ($n_2^*(z_1) = N_{\max}$): the decision is rationalized by any cost function small enough to favor increasing the sample size until the cap. Formally,
	\[
	\gamma^*(z_1) \leq \left. \frac{\partial CP_{\theta}(z_1, n_2)}{\partial n_2} \right|_{n_2 = N_{\max}}.
	\]
	\item \textbf{Minimum boundary} ($n_2^*(z_1) = n_{\min}$): the decision is rationalized by any cost function large enough to discourage increasing the sample size. Formally,
	\[
	\gamma^*(z_1) \geq \left. \frac{\partial CP_{\theta}(z_1, n_2)}{\partial n_2} \right|_{n_2 = n_{\min}}.
	\]
\end{itemize}
In these boundary cases, a cost function consistent with the rule still exists, but it is no longer uniquely defined.

\begin{remark}
	Denote the marginal gain in power (MGP) as the first partial derivative of the conditional power with respect to the sample size: $MGP(z_1, n_2) = \frac{\partial CP_{\theta}(z_1, n_2)}{\partial n_2}$. The strictly concave assumption on $CP_{\theta}(z_1, n_2) $ is equivalent to diminishing returns property of MGP, that is, $MGP(z_1, n_2)$ is a strictly monotonically decreasing function of $n_2$.  It implies that each additional participant provides less increment on the conditional power than the one before.
\end{remark}

\begin{remark}
 For the purpose of comparing SSR rules, our primary interest lies in their behavior within the promising zone. Hence, it suffices to establish a one-to-one correspondence between SSR rules and cost functions restricted to this region. Accordingly, it is sufficient for Theorem~\ref{ssr_rep} that the objective function be strictly concave with respect to $n_2$ for $z_1$ within the promising zone.
\end{remark}

\begin{remark}
The assumptions underlying Theorem~\ref{ssr_rep} can be relaxed. By employing tools such as subgradients, the requirements of differentiability and strict concavity may be replaced with considerably weaker conditions, namely upper semi-continuity and concavity. Under these weaker assumptions, one can define equivalence classes of SSR rules, allowing for discrete sample size functions $n_2(z_1)$ or non-normal endpoints. A formal development of this generalized version of the theorem, however, lies beyond the scope of the present work.
\end{remark}

The ability to reverse-engineer a SSR rule to find its implied cost function provides a powerful auditing tool for evaluating proposed or existing SSR designs. An analysis might reveal that a seemingly sensible rule implies a $\gamma$ function that is erratic, or counter-intuitive. For instance, a rule might implicitly value an additional participant more highly in a moderately promising zone than in a very promising zone, a behavior that may be difficult to justify rationally. Identifying such anomalies can serve as a critical red flag, indicating that the rule could lead to inefficient or logically inconsistent resource allocation under certain interim outcomes. This analytical capability allows for a deeper and more critical appraisal of SSR strategies, moving beyond surface-level performance metrics like average power and sample size to scrutinize the logical coherence of the underlying decision-making process.

Beyond retrospective evaluation, the representation theorem offers a constructive paradigm: sponsors may specify a rational cost function that reflects their risk-benefit preferences and derive the corresponding SSR rule. This approach enhances transparency and provides regulators with a principled justification that goes beyond operating characteristics to articulate the rationale underlying the design.
\subsection{Two sufficient conditions for strict concavity of $CP_{\theta}(z_1, n_2)$}

In this section,  we provide two sufficient conditions under which the conditional power $CP_{\theta}(z_1, n_2)$ is strictly concave in $n_2$.  
The first condition is based on a lower bound for $CP_{\theta}(z_1, n_2)$: whenever the conditional power exceeds this bound, it is strictly concave with respect to $n_2$.  
The second condition involves a lower bound on $z_1$: for values of $z_1$ above this threshold, $CP_{\theta}(z_1, n_2)$ is strictly concave in $n_2$ regardless of the specific value of $n_2$.  
Both conditions are straightforward to compute and practically useful for verifying the assumption required in Theorem~\ref{ssr_rep}. We introduce the notation $K = \theta / (2\sigma)$ and 
$
c = \frac{C_{\text{crit}} - w_1 z_1}{w_2},
$
where $w_1$ and $w_2$ are the weights in the inverse normal combination test, and $C_{\text{crit}}$ denotes the critical value corresponding to a given significance level $\alpha$.  
With this notation, we obtain the following lemma.

\begin{lemma}\label{suff_concave}
Assume normality. 	If 
	\[
	CP_{\theta}(z_1, n_2(z_1)) > \Phi\!\left(-\frac{1}{K \sqrt{\,\\NM - n_1\,}} \right)
	\]
	or
	\[
	z_1 > \max_{n_2 \in \admiss} \left[ \frac{C_{\text{crit}} - w_2\!\left(K \sqrt{\,n_2(z_1) - n_1\,} - \frac{1}{K \sqrt{\,n_2(z_1) - n_1\,}} \right)}{w_1} \right],
	\]
	then $CP_{\theta}(z_1, n_2(z_1))$ is strictly concave in $n_2(z_1)$.
\end{lemma}

\begin{proof}
	Let
	\[
	A(n_2(z_1)) = K \sqrt{\,n_2(z_1) - n_1\,} - c, 
	\quad \text{where } c = \tfrac{C_{\text{crit}} - w_1 z_1}{w_2}.
	\]
	The second derivative of $CP_{\theta}(z_1, n_2(z_1))$ with respect to $n_2(z_1)$ is
	\[
	\frac{\partial^2 CP_{\theta}(z_1, n_2(z_1))}{\partial n_2(z_1)^2} 
	= -A(n_2(z_1)) \, \phi(A(n_2(z_1))) \!\left(\frac{d A(n_2(z_1))}{d n_2(z_1)} \right)^2 
	+ \phi(A(n_2(z_1))) \frac{d^2 A(n_2(z_1))}{d n_2(z_1)^2}.
	\]
	
	Strict concavity requires
	\[
	\frac{\partial^2 CP_{\theta}(z_1, n_2(z_1))}{\partial n_2(z_1)^2} < 0.
	\]
	Noting that $\phi(\cdot) > 0$ and
	\[
	\frac{d^2 A(n_2(z_1))}{d n_2(z_1)^2} = -\frac{K}{4}\,(n_2(z_1)-n_1)^{-3/2} < 0,
	\]
	it suffices that
	\[
	A(n_2(z_1)) > -\frac{1}{K \sqrt{\,n_2(z_1) - n_1\,}}.
	\]
	This inequality is equivalent to either
	\[
	CP_{\theta}(z_1, n_2(z_1)) > \Phi\!\left(-\frac{1}{K \sqrt{\,\\NM - n_1\,}} \right)
	\]
	or
	\[
	z_1 > \max_{n_2 \in \admiss} \left[ \frac{C_{\text{crit}} - w_2\!\left(K \sqrt{\,n_2(z_1) - n_1\,} - \tfrac{1}{K \sqrt{\,n_2(z_1) - n_1\,}} \right)}{w_1} \right].
	\]
	Hence, either condition guarantees strict concavity.
\end{proof}

\subsection{An usage example of the representation theorem}
In this section, we illustrate the utility of Theorem~\ref{ssr_rep} by analyzing the ``Constrained Promising Zone" (CPZ) rule introduced in~\cite{hsiao2019optimal}. Assuming $\theta = 0.29$, the CPZ rule is defined as the solution to the following constrained optimization problem:

\begin{enumerate}
	\item[] \textbf{Objective:} Maximize $CP_{\theta}(z_1, n_2(z_1))$ with respect to $n_2(z_1)$ for all $z_1$, subject to
	\item[] \textbf{Constraint 1:} $280 \leq n_2(z_1) \leq 420$,
	\item[] \textbf{Constraint 2:} $CP_{\theta}(z_1, n_2(z_1)) \geq 0.8$,
	\item[] \textbf{Constraint 3:} $CP_{\theta}(z_1, n_2(z_1)) \leq 0.9$.
\end{enumerate}

On the right panel of Figure~\ref{fig:ex5_cost_revised}, the CPZ rule increases the sample size when the interim statistic $z_1$ falls within an interval, which they identify as $[1.187, 2.338]$. By lemma~\ref{suff_concave}, we can show that the conditional power is differentiable and strictly concave with respect to $n_2$ for any $z_1 > 0.474$. Therefore, we can deduce the implied cost function, $\gamma_{CPZ}(z_1)$, in three regions:

\begin{enumerate}
	\item \textbf{Lower Promising Zone ($z_1 \in [1.187, 1.627]$):} For this interval of $z_1$, its right endpoint $1.627$ will be denoted as the turning point of the CPZ design. Within the lower promising zone, the sample size is increased to the maximum, $n_2 = 420$.  This implies that the marginal benefit of adding another participant exceeds the cost, even at the maximum sample size. Thus, the cost function is bounded by:
	$$
	\gamma_{CPZ}(z_1) \le \left. \frac{\partial CP_{\theta}(z_1, n_2)}{\partial n_2} \right|_{n_2 = 420}
	$$
	
	\item \textbf{Upper Promising Zone ($z_1 \in (1.627, 2.338]$ ):} The sample size is set to an intermediate value $\tilde{n}_2(z_1)$ such that the conditional power hits a ceiling of 0.9. Here, the marginal benefit equals the marginal cost, uniquely defining the cost function:
	$$
	\gamma_{CPZ}(z_1) = \left. \frac{\partial CP_{\theta}(z_1, n_2)}{\partial n_2} \right|_{n_2 = \tilde{n}_2(z_1)}
	$$
	
	\item \textbf{Outside the Promising Zone:} The sample size is kept at the minimum, $n_2 = 280$. This implies that the marginal cost of adding a participant is greater than or equal to the marginal benefit:
	$$
	\gamma_{CPZ}(z_1) \ge \left. \frac{\partial CP_{\theta}(z_1, n_2)}{\partial n_2} \right|_{n_2 = 280}
	$$
\end{enumerate}

Figure~\ref{fig:ex5_cost_revised} compares the implied cost function and the resulting sample size rule of the CPZ method with the JT rule. The plots reveal that while both rules operate in the same region, their underlying rationales differ significantly. The variable nature of $\gamma_{CPZ}(z_1)$ shows the CPZ rule to be an adaptive strategist, valuing power differently based on the interim outcome. It assigns the lowest cost when $z_1$ equals to the turning point. In contrast, the more stable cost implied by the JT rule reflects a greater focus on unconditional optimization. Even though the final sample size rules for the CPZ and JT designs can appear superficially similar, their implied cost functions show they are driven by fundamentally different philosophies. This example shows that how to use Theorem~\ref{ssr_rep} to reveal a sample size re-estimation rule's underlying philosophy.

\begin{figure}[h!]
	\centering
	\includegraphics[width=0.48\textwidth]{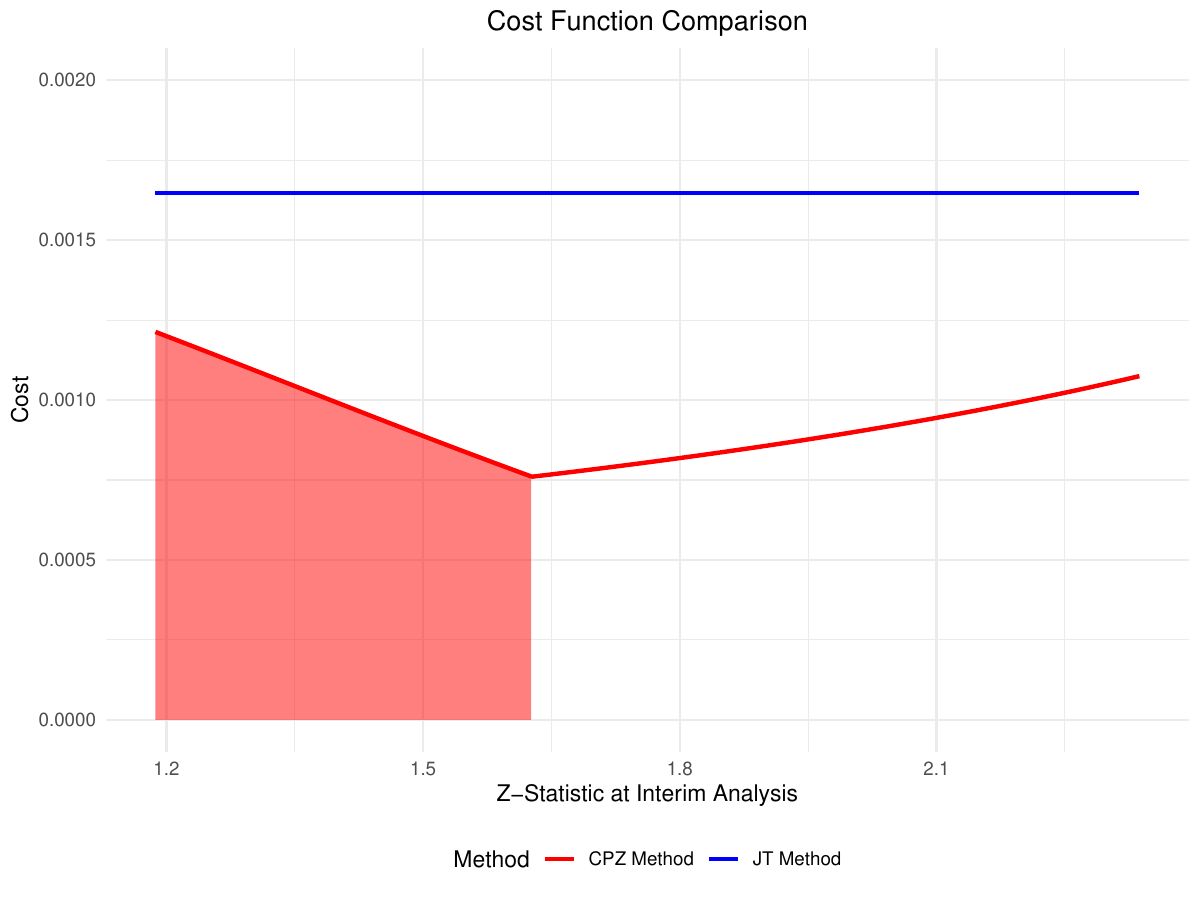}
	\includegraphics[width=0.48\textwidth]{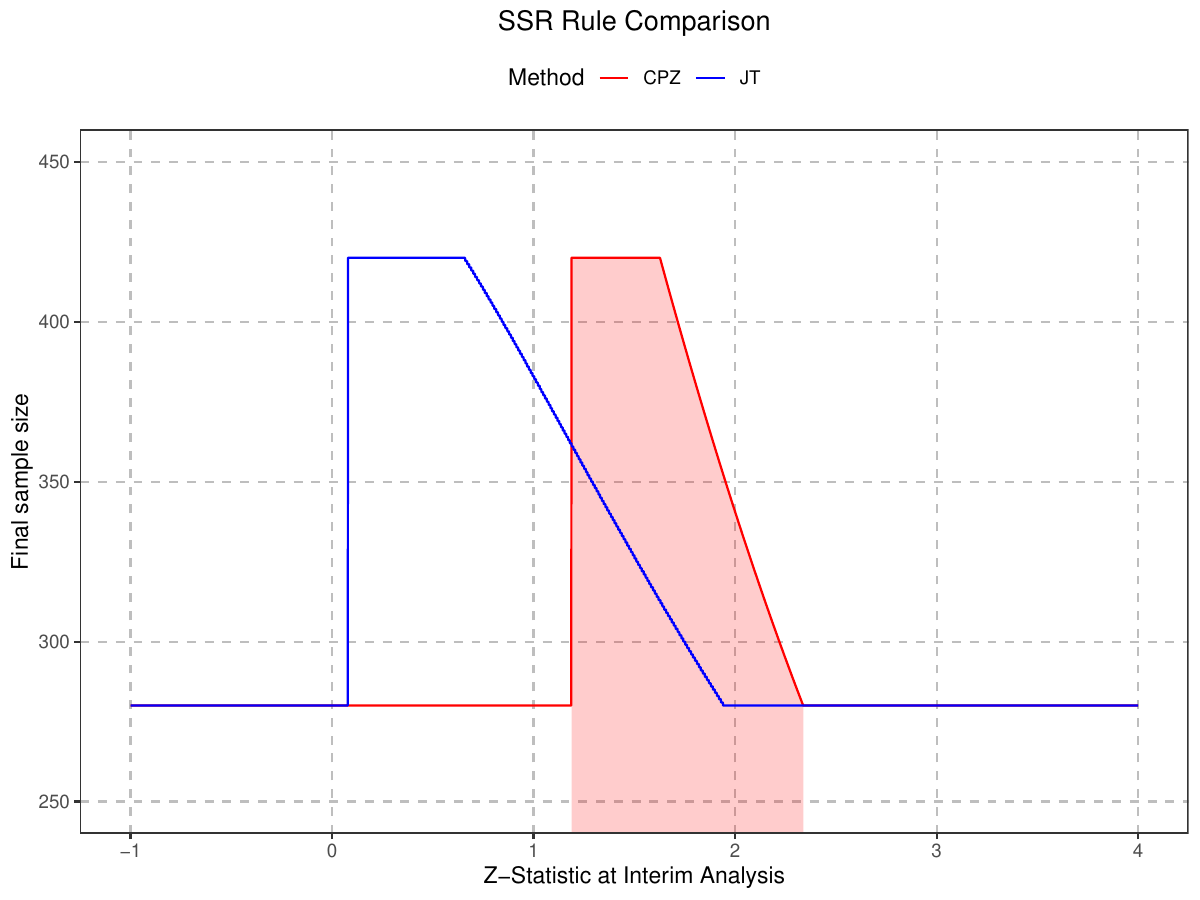}
	\caption{A comparison of the JT and CPZ rules through their implied cost functions (left) and resulting sample size rules (right).}
	\label{fig:ex5_cost_revised}
\end{figure}

\section{Simulation Studies}\label{sec:simulation}

To evaluate our proposed method, we adopt the same simulation framework used in previous studies by ~\cite{mehta2011adaptive} and~\cite{jennison2015adaptive}. We consider a Phase 3 clinical trial for a new schizophrenia treatment against an active control. The primary endpoint is assumed to be normally distributed with a known variance of $\sigma^2 = 7.5^2$.  The treatment difference, denoted by $\theta$  is the parameter of interest. An initial total sample size $n_2 = 442$  (with a 1:1 allocation ratio)  was  planned to detect an effect size of  $\theta = 1.6$. Thus, we assume $\Theta_0 = \{0\}$ and $\Theta_1 = \{1.6\}$. 

For all simulations, the type I error for all rules are maintained by applying combination test statistics (weighted inverse normal) described in section 5 of~\cite{jennison2015adaptive}.

\subsection{Comparing $n_2^{JT}$ with $n_2^{LR}$}
An interim analysis is performed after collecting data from  $n_1 = 208$ participants. The final adjusted sample size is constrained to be within the range $[442, 884]$. We compare the performance of our likelihood-ratio based rule $n_2^{LR}$  against the Jennison-Turnbull method $n_2^{JT}$.  For the JT method, we use cost parameter $\gamma = 0.25/(4\sigma^2)$, as specified in the original paper by~\cite{jennison2015adaptive}.  To provide a fair comparison, the Lagrange multipliers for our $n_2^{LR}$ rule, $\lambda_1 = 0.65\gamma$ and $\lambda_2 = 0.62\gamma$ are calibrated to match the average sample size under the alternative hypothesis  $\theta = 1.6$ of the JT method. Under this setup, the unconditional power of the two rules is nearly identical, at $64.3\%$ and $64\%$, respectively. As shown in Figure~\ref{cp_ex1}, the JT method tends to increase the sample size for relatively small interim test statistics $z_1$. This behavior carries a higher risk of wasting resources if the null hypothesis is true.  In contrast, $n_2^{LR}(z_1)$ accounts for this risk by adjusting the cost function with a likelihood ratio. This adjustment postpones aggressive sample size increases until a more promising $z_1$
is observed, thus providing a more conservative and resource-efficient approach.

\begin{figure}[h!]
	\includegraphics[width=.5\textwidth]{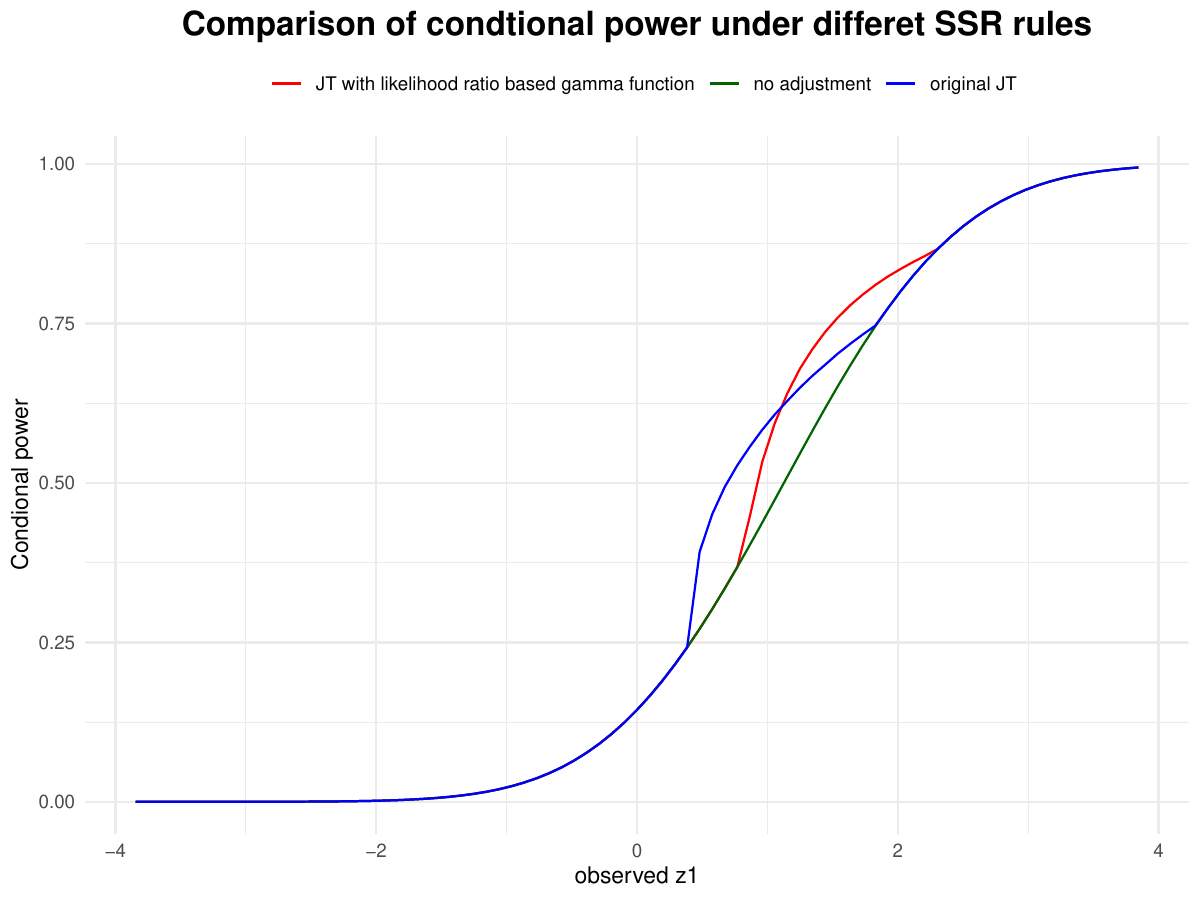}
		\includegraphics[width=.5\textwidth]{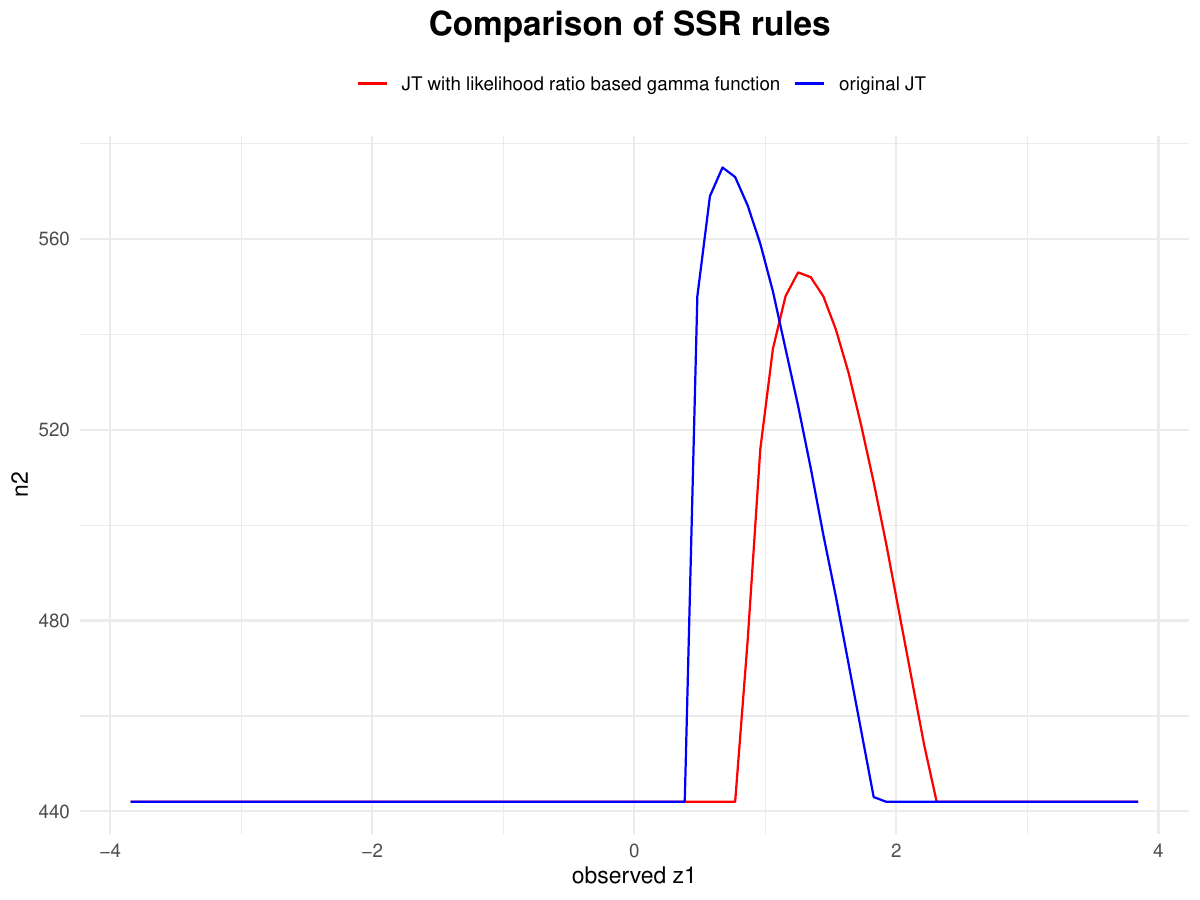}
		\caption{Comparison of $n_2^{JT}$ and $n_2^{LR}$.}
	\label{cp_ex1}
\end{figure}


\subsection{Effect of interim analysis timing on $n_2^{JT}$ and $n_2^{LR}$}

In this section, we keep $\gamma = 0.25/(4\sigma^2)$, $\lambda_1 = 0.65\gamma$ and $\lambda_2  = 0.62\gamma$ fixed while varying the interim sample size from 80 to 200 to  examine how the timing of the interim analysis affects the behavior of both methods. As shown in figure~\ref{ex2_JT}, these two methods exhibit completely different behavior. For an earlier interim analysis (smaller $n_1$), $n_2^{JT}(z_1)$  tends to increase the sample size for a weaker observed interim statistic. This is a high-risk strategy. On the contrary, due to the adjustment from the likelihood ratio, $n_2^{LR} (z_1)$ makes more conservative decisions when the interim analysis is conducted with a smaller sample size.

\begin{figure}[h!]
	\includegraphics[width=.5\textwidth]{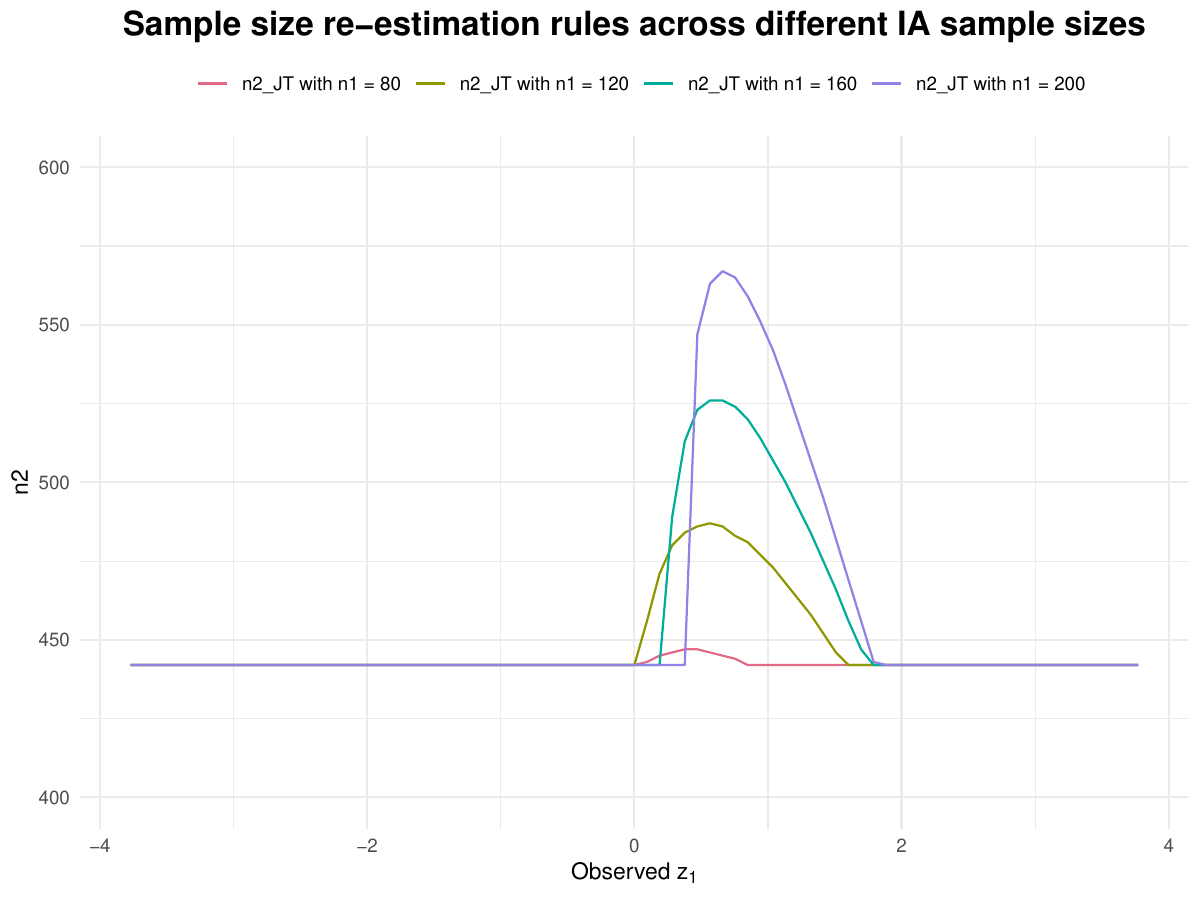}
		\includegraphics[width=.5\textwidth]{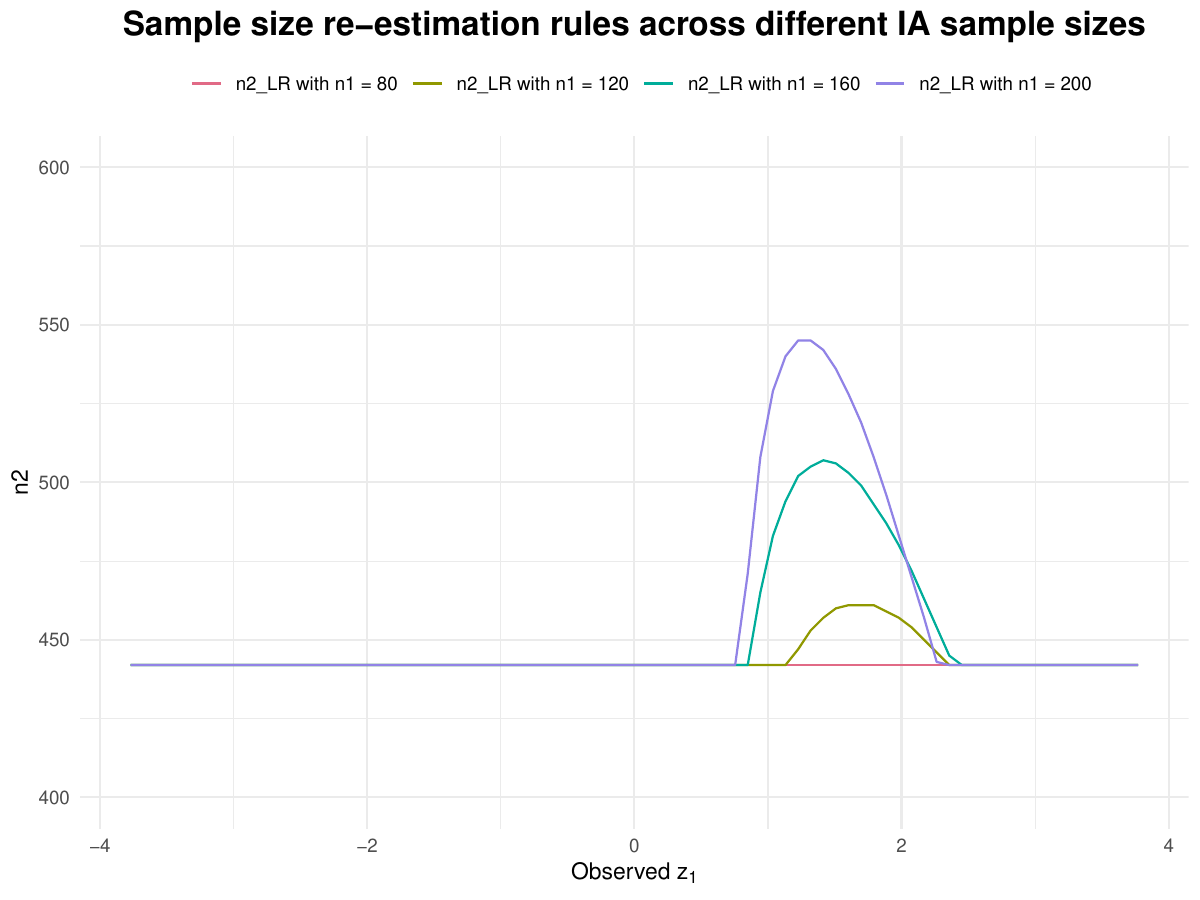}
		\caption{Sample size re-estimation rules $n_2^{JT}(z_1)$ and $n_2^{LR}(z_1)$ across different IA time points.}
	\label{ex2_JT}
\end{figure}


\subsection{ROI based sample size re-estimation}
In this section, we demonstrate the ability of our ROI-based sample size re-estimation rule, $n_2^{ROI}(z_1)$, to adapt to different cost-benefit profiles.  The method is based on a two-state prior for the null and alternative hypotheses, with two distinct sets of prior probabilities considered for illustrative purposes.
\begin{enumerate}
	\item[a.] $\pi_0 = \pi_1 = 1/2$
	\item[b.] $\pi_0 = \frac{2}{3}$ and $\pi_1 = \frac{1}{3}$
\end{enumerate}
Based on our proposed framework, the dynamic cost function for each case is given by:
\begin{equation*}
\gamma(z_1)^a = \frac{c}{V} \left(1 + \frac{f_0(z_1)}{f_{\theta}(z_1)} \right) \,,
\end{equation*}
and
\begin{equation*}
\gamma(z_1)^b = \frac{c}{V} \left(1 + \frac{2f_0(z_1)}{f_{\theta}(z_1)} \right) \,.
\end{equation*}

Here, $c$ represents the cost of adding a single participant, and $V$ is the total return if the treatment is successful. We investigate the behavior of $n_2^{ROI_a}(z_1)$ and $n_2^{ROI_b}(z_1)$ under various combinations of $(c, V)$ that reflect different ROI scenarios.

For this analysis, we assume a fixed total return $V = \$100 \text{ million}$. We then vary the cost per participant, $c$, from $\$40,000$ to $\$100,000$ to observe how the re-estimation rule's behavior changes. The results, as shown in Figure~\ref{ex3_ROI}, reveal a clear pattern. For a given cost per participant,  the rule $n_2^{ROI_a}(z_1)$  is consistently more aggressive in adding samples than  $n_2^{ROI_b}(z_1)$,   which is an expected outcome given the larger prior probability assigned to the alternative hypothesis $\Theta_1$ in case a. Across different cost values,  both rules tend to recommend more aggressive sample size increases when the cost per participant is low and become more conservative as the cost increases. This demonstrates how the ROI-based rule allows for the direct integration of economic considerations into the design of a sample size re-estimation rule, providing a flexible and economically rational approach to clinical trial design.

\begin{figure}[h!]
	\includegraphics[width=.5\textwidth]{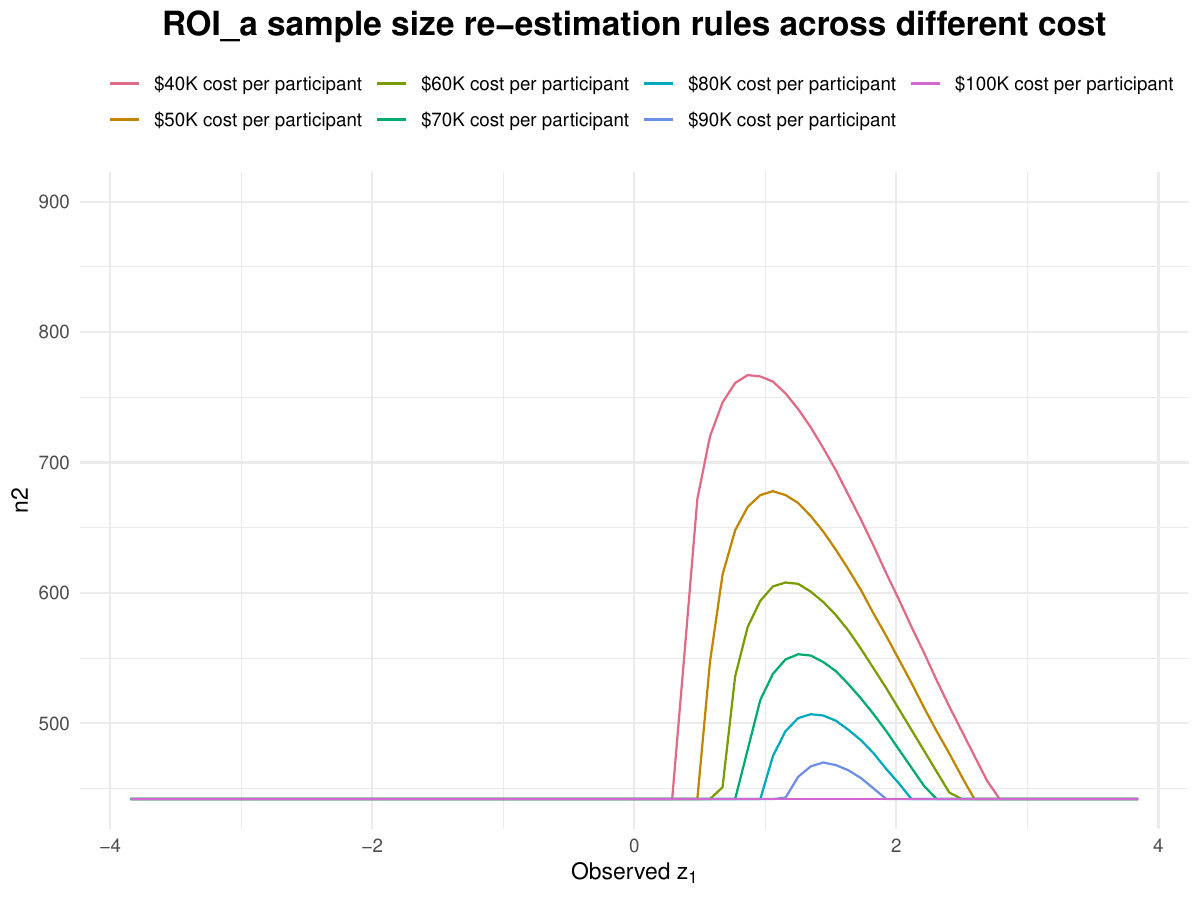}
	\includegraphics[width=.5\textwidth]{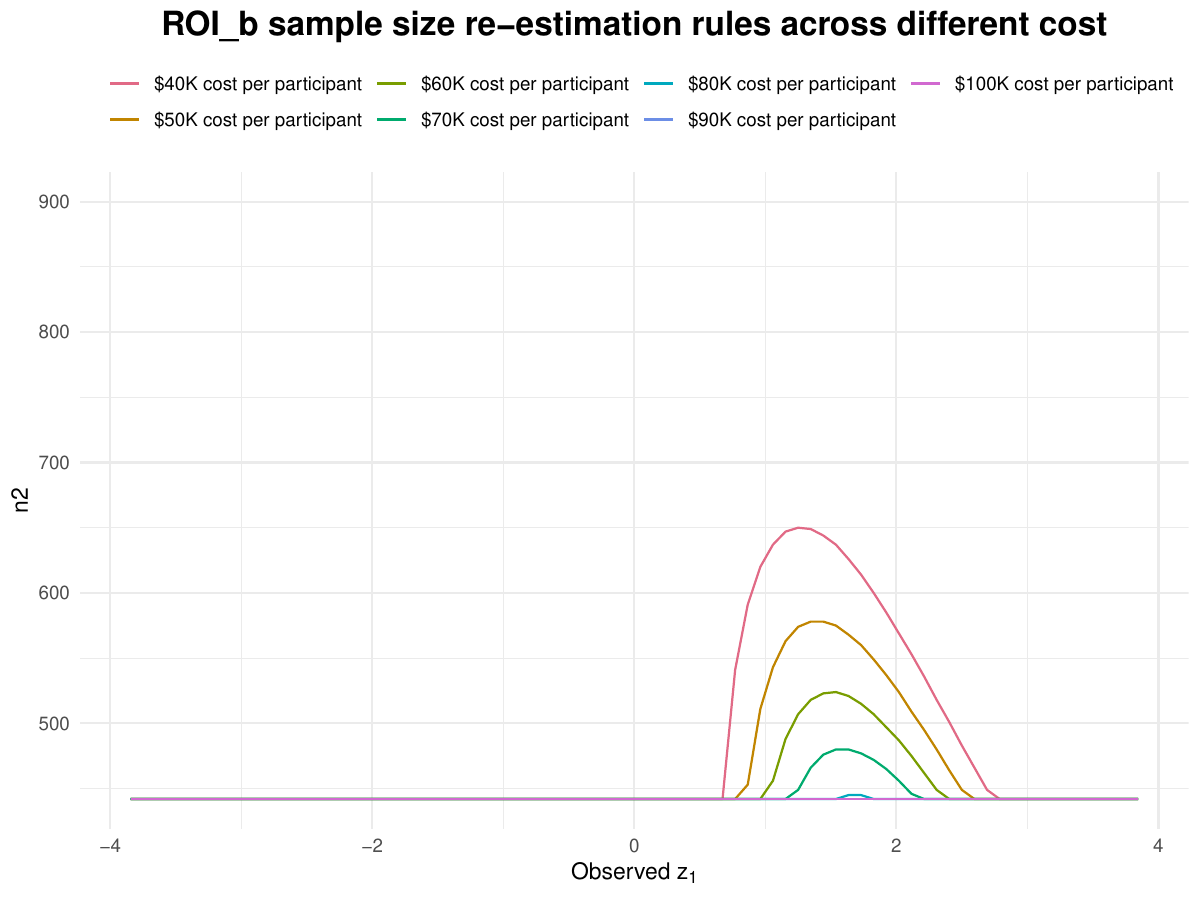}
	\caption{$n_2^{ROI_a}(z_1)$ and $n_2^{ROI_b}(z_1)$ across different cost amounts per participant.}
	\label{ex3_ROI}
\end{figure}

%
%
%
\section{Conclusions}\label{sec:end}

This paper introduces a novel and generalized framework for sample size re-estimation under dynamic costs, which adapts to the strength of interim evidence. From this framework, we developed two novel rules: the Pareto-optimal {likelihood-ratio based rule}, $n_{2}^{LR}(z_1)$, which enhances statistical efficiency by minimizing the expected sample size under the null hypothesis , and the {return on investment based rule}, $n_{2}^{ROI}(z_1)$, which directly integrates economic considerations into the re-estimation process.

Beyond these specific rules, we established a representation theorem that unifies the broader field of SSR rules within our dynamic cost framework. This theorem provides a powerful new lens for the critical appraisal of adaptive designs for sample size re-estimation. It serves a dual purpose: first, as an auditing tool to reverse-engineer a SSR rule and reveal its implicit cost structure; and second, as a new paradigm for designing a SSR rule, where sponsors can prospectively define a rational cost function that reflects their true risk tolerance and strategic goals.

In conclusion, our work provides a more principled and flexible foundation for adaptive sample size re-estimation. By shifting the focus from procedural rules to the underlying economic and statistical rationale, the dynamic cost framework supports the development of more efficient, transparent, and justifiable SSR rules.

\bibliographystyle{ims-mine}
\bibliography{isregen}

\begin{thebibliography}{8}
\expandafter\ifx\csname natexlab\endcsname\relax\def\natexlab#1{#1}\fi
\expandafter\ifx\csname url\endcsname\relax
  \def\url#1{\texttt{#1}}\fi
\expandafter\ifx\csname urlprefix\endcsname\relax\def\urlprefix{URL }\fi

\bibitem[{Friede and Kieser(2013)}]{friede2013blinded}
\text{Friede, T.} and \text{Kieser, M.} (2013).
\newblock Blinded sample size re-estimation in superiority and noninferiority
  trials: bias versus variance in variance estimation.
\newblock \textit{Pharmaceutical Statistics} \textbf{12} 141--146.

\bibitem[{Hsiao et~al.(2019)Hsiao, Liu and Mehta}]{hsiao2019optimal}
\text{Hsiao, S.~T.}, \text{Liu, L.} and \text{Mehta, C.~R.} (2019).
\newblock Optimal promising zone designs.
\newblock \textit{Biometrical Journal} \textbf{61} 1175--1186.

\bibitem[{Jennison and Turnbull(2015)}]{jennison2015adaptive}
\text{Jennison, C.} and \text{Turnbull, B.~W.} (2015).
\newblock Adaptive sample size modification in clinical trials: start small
  then ask for more?
\newblock \textit{Statistics in medicine} \textbf{34} 3793--3810.

\bibitem[{Mehta et~al.(2022)Mehta, Bhingare, Liu and
  Senchaudhuri}]{mehta2022optimal}
\text{Mehta, C.}, \text{Bhingare, A.}, \text{Liu, L.} and \text{Senchaudhuri,
  P.} (2022).
\newblock Optimal adaptive promising zone designs.
\newblock \textit{Statistics in Medicine} \textbf{41} 1950--1970.

\bibitem[{Mehta and Pocock(2011)}]{mehta2011adaptive}
\text{Mehta, C.~R.} and \text{Pocock, S.~J.} (2011).
\newblock Adaptive increase in sample size when interim results are promising:
  a practical guide with examples.
\newblock \textit{Statistics in medicine} \textbf{30} 3267--3284.

\bibitem[{Pilz et~al.(2021)Pilz, Kunzmann, Herrmann, Rauch and
  Kieser}]{pilz2021optimal}
\text{Pilz, M.}, \text{Kunzmann, K.}, \text{Herrmann, C.}, \text{Rauch, G.} and
  \text{Kieser, M.} (2021).
\newblock Optimal planning of adaptive two-stage designs.
\newblock \textit{Statistics in Medicine} \textbf{40} 3196--3213.

\bibitem[{Pocock(1977)}]{pocock1977group}
\text{Pocock, S.~J.} (1977).
\newblock Group sequential methods in the design and analysis of clinical
  trials.
\newblock \textit{Biometrika} \textbf{64} 191--199.

\bibitem[{Proschan(2005)}]{proschan2005two}
\text{Proschan, M.~A.} (2005).
\newblock Two-stage sample size re-estimation based on a nuisance parameter: a
  review.
\newblock \textit{Journal of biopharmaceutical statistics} \textbf{15}
  559--574.

\end{thebibliography}

\end{document}